\documentclass[aps,pra,superscriptaddress,notitlepage,twocolumn]{revtex4-1}
\usepackage[utf8]{inputenc}
\usepackage{braket}
\usepackage{amsmath,amsfonts}
\usepackage{graphicx}
\usepackage{bm}
\usepackage{amsthm}
\usepackage[colorlinks]{hyperref}
\usepackage{xcolor}
\usepackage{comment}
\usepackage{mathtools}
\newtheorem{theorem}{Theorem}
\newtheorem{coro}{Corollary}
\newtheorem{lemma}{Lemma}
\newtheorem{defi}{Definition}

\newcommand{\bbC}{\mathbb{C}}
\newcommand{\ba}{\begin{eqnarray}}
\newcommand{\ea}{\end{eqnarray}}
\newcommand{\ban}{\begin{eqnarray*}}
\newcommand{\ean}{\end{eqnarray*}}
\makeatletter
\newcommand{\RN}[1]{\expandafter\@slowromancap\romannumeral #1@}
\makeatother

\newcommand{\adam}[1]{\textcolor{gray}{#1}}

\usepackage{soul}

\newcommand{\specialcell}[2][c]{%
  \begin{tabular}[#1]{@{}c@{}}#2\end{tabular}}

\begin{document}
\title{Absolutely entangled sets of pure states for bipartitions and multipartitions}

\author{Baichu Yu}
\affiliation{Centre for Quantum Technologies, National University of Singapore, 3 Science Drive 2, Singapore 117543}

\author{Pooja Jayachandran}
\affiliation{Centre for Quantum Technologies, National University of Singapore, 3 Science Drive 2, Singapore 117543}

\author{Adam Burchardt}
\affiliation{Jagiellonian University, Marian Smoluchowski Institute for Physics, \L ojasiewicza 11, 30-348 Krak\'{o}w, Poland}

\author{Yu Cai}
\affiliation{Department of Applied Physics, University of Geneva, Geneva, Switzerland}


\author{Nicolas Brunner}
\affiliation{Department of Applied Physics, University of Geneva, Geneva, Switzerland}

\author{Valerio Scarani}
\affiliation{Centre for Quantum Technologies, National University of Singapore, 3 Science Drive 2, Singapore 117543}
\affiliation{Department of Physics, National University of Singapore, 2 Science Drive 3, Singapore 117542}


\begin{abstract}
A set of quantum states is said to be absolutely entangled, when at least one state in the set remains entangled for any definition of subsystems, i.e. for any choice of the global reference frame. In this work we investigate the properties of absolutey entangled sets (AES) of pure quantum states. For the case of a two-qubit system, we present a sufficient condition to detect an AES, and use it to construct families of $N$ states such that $N-3$ (the maximal possible number) remain entangled for any definition of subsystems. For a general bipartition $d=d_1d_2$, we prove that sets of $N>\left\lfloor{(d_{1}+1)(d_{2}+1)/2}\right \rfloor$ states are AES with Haar measure 1. Then, we define AES for multipartitions. We derive a general lower bound on the number of states in an AES for a given multipartition, and also construct explicit examples. In particular, we exhibit an AES with respect to any possible multi-partitioning of the total system.
\end{abstract}

\maketitle

\section{Introduction}

The notion of entanglement relies on partitioning a larger system into subsystems. The definition of subsystems, however, usually relies on operational convenience \cite{zanardi2001, zanardi2004, harshman2011}. For example, a very common scenario to display entanglement considers that a general system is spatially separated into two distant parts, held by Alice and Bob respectively. In this case it is natural to take the two local systems as the subsystems, since general joint operations are not allowed.  However, there are cases that lack such a natural choice of partition. The textbook treatment of the hydrogen atom is a case in point: if we choose the proton and electron as subsystems, then there is entanglement, since the proton can be delocalised, while the electron is next to it. However, if we choose the subsystems to be the centre of mass (CM) and relative coordinate, the system is in a separable form. Hence, a system entangled with respect to a certain definition of subsystems and corresponding local operations can be made separable by redefining the subsystems. And if multiple such definitions exist, the meaning of entanglement is blurred. 

More formally, the definition of subsystems relies on the choice of a reference frame. A change of reference frame corresponds to applying a unitary operation on the total system. It is then clear that any single quantum state $\ket{\phi}\in\bbC^{d}$ can always be transformed into product state by a global unitary matrix $U\in L^{d}$, where $L^{d}$ is the space of linear operators on $\bbC^{d}$.  

Recently a twist on this problem was proposed by some of the authors here independently~\cite{lovitz2019decomposable,cai2020entanglement}, in different mathematical approach but with similar physical meaning. The main idea consists in considering sets of quantum states. There exist indeed sets of states, such that no unitary maps simultaneously all states in the set onto separable ones. In \cite{cai2020entanglement}, such sets are named as ``absolutely entangled sets'' (AES), and therefore always feature (at least) one entangled state, in all possible reference frames, i.e. for any possible definition of subsystems. 

The existence of AES can be intuitively predicted. Surely there is no unitary that simultaneously maps all states $\ket{\phi}\in\bbC^{d}$ to product ones. Examples of finite AES are provided by \textit{projective $2$-designs} on $\bbC^{d}$ with $d=d_1d_2$, which are sets of states that faithfully approximate the state space for any degree $2$ polynomial in the state coefficients \cite{Welch,klappenecker2005mutually}. Indeed, denoting by $\rho$ the reduced density matrix of the first subsystem, the average value of the purity $\left\langle  \text{Tr} \rho^2 \right\rangle$ evaluated over all vectors which form a $2$-design takes the same value as evaluated over all quantum states distributed according to the Haar measure. The latter is known to be $\left\langle  \text{Tr} \rho^2 \right\rangle=(d_1+d_2)/(d_1d_2+1)$ \cite{Lubkin_1978, Zyczkowski_2001}, different from one: thus, some states of the $2$-design must be entangled. A particularly famous example of such design is constituted by all vectors from the complete set of \textit{Mutually Unbiased Basis} (MUBs) \cite{Wooters_1989,WOOTTERS19871}. Two bases are said to be \textit{unbiased} if for a system prepared in an eigenstate of one basis, all outcomes of the measurement with respect to the other bases are predicted to occur with equal probability \cite{Bengtsson_2007}. 
Construction of a complete set of MUBs is known for any prime power dimension $d=p^s$ \cite{MUBalg,MUBs25} and it is also known that at most $d(p+1)$ vectors in such a set are separable \cite{LocalMubs,Wooters_1989,Czartowski_2018}. Therefore, any subset of $d (p+1)+1$ vectors from the complete set of MUBs in dimension $d=p^s$ forms an AES.

A natural question is then to find AES featuring only few states, and even minimal sets. This question was discussed in Ref.~\cite{lovitz2019decomposable,cai2020entanglement}. 
First, a general lower bound on the size of an AES was derived: for the case of pure states in $\bbC^{d}=\bbC^{d_1}\otimes\bbC^{d_2}$, one needs at least $\max{(d_1,d_2)}+2$ states~\cite{lovitz2019decomposable,cai2020entanglement}. Second, an explicit construction of AES featuring $d_1+d_2$ states was presented (this construction describes therefore a minimal set for $\min(d_1,d_2)=2$). Tools to quantify the amount of entanglement in an AES were also presented~\cite{cai2020entanglement}. In a subsequent paper, Li and Yung provided more examples of AES for $\bbC^{d}=\bbC^{d_1}\otimes\bbC^{d_2}$ case with $d$ and $d+1$ states \cite{li2020absolutely}. 

In this paper we present a broader exploration of the properties of AES of pure quantum states. After a quick review of previously known results, Section \ref{sec:bipartite} presents new results for bipartitions $\bbC^{d}=\bbC^{d_1}\otimes\bbC^{d_2}$. First, we give a sufficient condition for a set of four linearly two-qubit independent states to be AES (Theorem \ref{thm1}). This theorem becomes a tool to construct sets of $N$ states, $N-3$ of which remain entangled for any $U$ (Theorem \ref{thm2}). We also estimate the fraction of AES for sets drawn with the Haar measure, proving in particular that almost all sets are AES if one takes sufficiently many states (Theorem \ref{thm3new}). In section \ref{sec:multi}, we move beyond the bipartite problem, and give a definition of an AES for \textit{multipartitions}, i.e. $\bbC^{d}=\bbC^{d_1}\otimes\bbC^{d_2}\otimes...\otimes\bbC^{d_k}$. We derive a lower bound on the size of an AES given by $\max{(d_1,d_2,...,d_{k})}+2$ (Theorem \ref{thm3text}). We also construct an AES of $(d_1+d_2+...+d_{k}-k+2)$ states with respect to multipartitions $\bbC^{d}=\bbC^{d_1}\otimes\bbC^{d_2}\otimes...\otimes\bbC^{d_k}$ (Theorem \ref{thm4}). As a corollary follows the construction of a set that is AES with respect to every multi-partition of $\bbC^{d}$. Section \ref{sec:concl} sketches some open questions and future directions.

\section{Properties of bipartite AES}
\label{sec:bipartite}

We recall the notion of an absolutely entangled set (AES) of states for bipartitions of the Hilbert space proposed in \cite{cai2020entanglement}:
\begin{defi}
Consider a set of quantum states $\{\rho_1,...,\rho_K\}$ in a fixed Hilbert space $\bbC^d$ of non-prime dimension. The set is said to be absolutely entangled with respect to bipartitions into subsystems of fixed dimension $(d_1,d_2)$, if for every unitary $U \in SU(d)$, at least one state $U\rho_kU^\dagger$ is entangled with respect to $\bbC^{d_1} \otimes \bbC^{d_2}$.
\end{defi}

In \cite{cai2020entanglement}, it was shown that an arbitrary set of pure states would require a minimum of $N_{min}(d_1,d_2)=\max(d_1,d_2)+2$ states to form an AES: that is, for all sets of states less than $N_{min}$, there exists a bipartition in which all states are product. Several examples of AES were also presented. For instance, the set of $d_1+d_2$ states
\begin{align}\label{specialset}
\begin{split}
    \ket{\phi_1} =& \ket{\xi_1}, \\
    \ket{\phi_k} =& c\ket{\xi_{1}}+\sqrt{1-c^2}\ket{\xi_{k}},\;k=2,...,d_1+d_2\,,
    \end{split}
\end{align}
where $\langle \xi_{j}|\xi_k\rangle=\delta_{jk}$, is absolutely entangled for any value $c\in\big(\sqrt{\frac{(d_1-1)(d_2-1)}{d_1d_2}},1\big)$.

In this section we list additional results on AES for bipartitons. We focus on the simplest case of $d=4$ i.e.~$d_1=d_2=2$. There seems to be no obstacle of principle in extending similar analyses to larger dimensions, but the expressions are already rather cumbersome.

\subsection{Sufficient condition for AES of linearly independent states.}

In our previous work, we showed specific examples of AES involving four states in $d=4$. All those examples used linearly independent states. Here we present an extension of that proof, that provides a sufficient condition for an arbitrary set of linearly independent states to be AES.

\begin{theorem}\label{thm1}
Consider an orthonormal basis $\{\ket{\xi_i}\}_{i=1}^4$ of $d=4$ ($d_1=d_2=2$) dimensional Hilbert space $\mathbb{C}^4$ and four linearly independent states
\begin{align}\label{eqth11}
    \begin{split}
    \ket{\phi_1} &= \ket{\xi_1},\\
    \ket{\phi_2} &= c_{21}\ket{\xi_1} + c_{22}\ket{\xi_2},\\
    \ket{\phi_3} &= c_{31}\ket{\xi_1} + c_{32}\ket{\xi_2}+c_{33}\ket{\xi_3},\\
    \ket{\phi_4} &= c_{41}\ket{\xi_1} + c_{42}\ket{\xi_2}+c_{43}\ket{\xi_3}+c_{44}\ket{\xi_4}.
     \end{split}
\end{align}
The set of states $\{\ket{\phi_i}\}_{i=1}^4$ is absolutely entangled if 
\begin{align}
    c\adam{:= }\min_{i=2,3,4}|c_{i1}|>1-\frac{2}{L+1}\label{eqcriterion1}
    \adam{,}
\end{align} 
where $L$ is is a positive number presented in terms of coefficients $|c_{ij}|$ in \eqref{eqL}.
\end{theorem}
\noindent
Notice that values $c$ and $L$ depend on the ordering of the states $\ket{\phi_i}$, while AES property is independent of such ordering. It is enough that one permutation of the states fulfills \eqref{eqcriterion1} to guarantee that the set is AES.

\textit{Proof.} Suppose there exists a global unitary matrix $U$ that takes all four states $\ket{\phi_i}$ into product states simultaneously. Then, up to further local unitary transformations, we can always assume
\begin{align}\label{eqth13}
\begin{split}
    U\ket{\phi_1}&=\ket{11},\\
    U\ket{\phi_2}&=d_{21}\ket{11}+d_{22}\ket{12}+d_{23}\ket{21}+d_{24}\ket{22},\\
    U\ket{\phi_3}&=d_{31}\ket{11}+d_{32}\ket{12}+d_{33}\ket{21}+d_{34}\ket{22},\\
    U\ket{\phi_4}&=d_{41}\ket{11}+d_{42}\ket{12}+d_{43}\ket{21}+d_{44}\ket{22}.
\end{split}
\end{align}
We shall analyze how the matrix $U$ transforms basis vectors $\ket{\xi_{1}}$. 
Note that $U\ket{\xi_{1}}=\ket{11}$, whence the action on the other basis vectors reads
\begin{align}\label{eqth12}
\begin{split}
    U|\xi_{2}\rangle=b_{22}|12\rangle+b_{23}|21\rangle+b_{24}|22\rangle,\\
U|\xi_{3}\rangle=b_{32}|12\rangle+b_{33}|21\rangle+b_{34}|22\rangle,\\
U|\xi_{4}\rangle=b_{42}|12\rangle+b_{43}|21\rangle+b_{44}|22\rangle,
\end{split}
\end{align}
where
\begin{align}
U_{B}=\left(
 \begin{matrix}
   b_{22} & b_{23} & b_{24} \\
   b_{32} & b_{33} & b_{34} \\
   b_{42} & b_{43} & b_{44}
  \end{matrix}
  \right)
  \end{align}
is a unitary matrix. By identification of the coefficients of \eqref{eqth11}, \eqref{eqth13} and \eqref{eqth12}, one gets $d_{i1}=c_{i1}$ and 
\begin{align}\label{eqds}
    \begin{split}
    d_{2j}&=c_{22}\cdot b_{2j},\\
    d_{3j}&=c_{32}\cdot b_{2j}+c_{33}\cdot b_{3j},\\
    d_{4j}&=c_{42}\cdot b_{2j}+c_{43}\cdot b_{3j}+c_{44}\cdot b_{4j}
    \end{split}
    \end{align}
for $j=2,3,4$. Observe that three states $\{U\ket{\phi_i}\}_{i=2}^4$  in \eqref{eqth13} are product if and only if $d_{i1}d_{i4}=d_{i2}d_{i3}$. By imposing this condition together with the normalisation constraints on vectors $U\ket{\psi_i}$, it follows that the largest value that $|d_{i4}|$ can take (denoted $|d_{i4}|_{max}$) is attained when $|d_{i2}|=|d_{i3}|$. This results in the expression
\begin{equation}
|d_{i1}||d_{i4}|_{max}=(1-|d_{i1}|^2-|d_{i4}|_{max}^{2})/2,
\end{equation}
hence we have
\begin{equation}\label{eqdi4}
|d_{i4}|_{max}=1-|c_{i1}|,
\end{equation}
where for simplicity we use the notation $|d_{i1}|=|c_{i1}|$. Let
    \begin{equation}\label{defSi}
 S_i \,:=\,\sqrt{\frac{1-|c_{i1}|}{1+|c_{i1}|}}=\sqrt{\frac{2}{|c_{i1}|+1}-1},
    \end{equation}
\noindent 
combining \eqref{eqds} with \eqref{eqdi4} yields
 \ban
    |b_{24}|&\leq& S_{2}\,\leq S_{max},\\
    |b_{34}|&\leq&  \frac{S_{2}\,|c_{32}|+S_{3}\,\sqrt{1-|c_{31}|^2}}{|c_{33}|}\\
    &\leq& S_{max}
    \,\frac{|c_{32}|+\sqrt{1-|c_{31}|^2}}{|c_{33}|},\\
    |b_{44}|&\leq& S_{max} \big(|c_{42}|+|c_{43}|\frac{|c_{32}|+\sqrt{1-|c_{31}|^2}}{|c_{33}|}\\&&+\sqrt{1-|c_{41}|^2}\big)/|c_{44}|,
\ean where $S_{max}=\textrm{max}_{i}\{S_{i}\}=\sqrt{\frac{2}{c+1}-1}$. The first of above inequalities implies that $c_{22}\neq 0$. Note that this requires the linear independence of vectors $\ket{\phi_1}$ and $\ket{\phi_2}$. Finally, the normalization constraint: $|b_{24}|^2+|b_{34}|^2+|b_{44}|^2=1$, implies that a necessary condition for all states being product is $S_{max}^{2}\geq 1/L$ i.e.~$c\leq 1-\frac{2}{L+1}$, where 
\begin{align}
     \label{eqL}
  L=&1+\left(r_{32}+\sqrt{r_{32}^2+1} \right)^2+\\
  \nonumber
 &\left(r_{42}+r_{43}\,\Big(r_{32}+\sqrt{r_{32}^2+1}\Big)+\sqrt{1+r_{42}^2+r_{43}^2}\right)^2,
 \end{align}
and $r_{ij}=|c_{ij}|/|c_{ii}|$. Thus, if \eqref{eqcriterion1} holds, the set of states is an AES. \qed

Notice that the expression of $L$ depends only on the coefficients of $\ket{\phi_3}$ and $\ket{\phi_4}$; as for $\ket{\phi_2}$, we mentioned in the proof that we need $c_{22}\neq 0$. Since $L=\infty$ when either $c_{33}$ or $c_{44}$ is zero, linear-independence is required for the bound \eqref{eqcriterion1} to be non-trivial. Also, given a set of linearly independent states that is not AES, one can generate a AES set by playing with the scaling coefficients $r_{ij}$ (the normalisation can be compensated with $c_{i1}$, that does not enter the expression).

\subsection{Sets of linearly dependent states}

So far we have discussed AES that span the whole space. Sets of linearly-dependent states are of measure zero among the sets with the same number of states; nonetheless, one can ask whether four states that do not span the whole space can be an AES for $d=4$.

Obviously, sets of $n$ states that span a $m$ dimensional subspace cannot be absolutely entangled if $m< \max(d_1,d_2)$. Indeed, a global unitary matrix can map the spanned subspace onto the subspaces spanned either by vectors $\ket{j}\otimes\ket{1}$ (if $d_1\geq d_2$) or by vectors $\ket{1}\otimes\ket{j}$ (if $d_2\geq d_1$). Therefore, any AES in dimension $d=4$ spans at least a $m=3$ dimensional subspace. On generating random sets of four linearly dependent states of the form,
\begin{align}
    \begin{split}
    \ket{\phi_1} &= \ket{\xi_1},\\
    \ket{\phi_2} &= a_1\ket{\xi_1} + a_2\ket{\xi_2},\\
    \ket{\phi_3} &= b_1\ket{\xi_1}+b_2\ket{\xi_2} +b_3\ket{\xi_3},\\
    \ket{\phi_4} &=c_1\ket{\xi_1}+c_2\ket{\xi_2} + c_3\ket{\xi_3},\\
     \end{split}
\end{align}
and running a heuristic optimisation using Matlab's \texttt{fminunc} function on them, we find that some sets attain a total entanglement entropy of the order of $10^{-1}$ \footnote{For example, we get a total entropy 0.13 when $a_1=0.2922 - 0.0351i$, $a_2=-0.7764 + 0.5573i$, $b_1=-0.0595 + 0.4964i$, $b_2=0.5150 + 0.2846i$, $b_3=-0.6334 - 0.0518i$, $c_1=0.6996 + 0.1303i$, $c_2=0.0494 + 0.0451i$, $c_3= -0.2643 - 0.6475i$}. Even though this is not analytic proof, it provides some evidence that there exist sets of four linearly dependent states that are AES.

\subsection{AES with maximal number of entangled states.}

The definition of AES requires that \textit{at least one} state remains entangled in any bipartition. But \textit{how many} states can remain entangled in all bipartitions? Since any set of $\max(d_1,d_2)+1$ states can be transformed into the separable states by a global unitary matrix, it follows that in a set of $N$ states, at most $N-[\max(d_1,d_2)+1]$ remains entangled with respect to any bipartition. We show that this bound is tight for sets in dimension $d=4$, i.e. $d_1=d_2=2$.

\begin{theorem}\label{thm2}
For any number $N$, there exist sets of $N$ two-qubit states such that for any bipartition, at least $N-3$ of them remains entangled.
\end{theorem}

We shall construct a set of $N$ states with the property that any of its four-element subsets constitutes AES. Thus, if one chooses any unitary matrix which transforms given three states to separable states, all other states remain entangled. 

We refer to Appendix \ref{appthm2} for the general construction of such set of states, here we show the example of $N=5$ states:
\begin{align}
    \begin{split}
    \ket{\phi_1} &= \ket{\xi_1},\\
    \ket{\phi_2} &= b\ket{\xi_1} + \sqrt{1-b^2}\ket{\xi_2},\\
    \ket{\phi_3} &= b\ket{\xi_1} + \sqrt{1-b^2}\ket{\xi_3},\\
    \ket{\phi_4} &= b\ket{\xi_1} + \sqrt{1-b^2}\ket{\xi_4},\\
    \ket{\phi_5} &= b\ket{\xi_1} +\sqrt{\frac{1-b^2}{3}}\big(\ket{\xi_2}+\ket{\xi_3}+\ket{\xi_4}\big)\,.
     \end{split}
\end{align} 
By symmetry of the above system of states, it is enough to check that the following three four-element subsets with subscript indices: $\{1,2,3,4\}$, $\{1,2,3,5\}$, $\{2,3,4,5\}$ form AES. We use the criterion presented in Theorem \ref{thm1}, checking it for all the permutations within those subsets. The critical values of $b$ for the three subsets are 0.5, 0.820, 0.762 respectively. Thus, for any value $b>0.820$, each four-element subset of the set $\{\ket{\phi_i} \}_{i=1}^5$ is an AES.

\subsection{Relative volume of AES.}

The next question we are going to address is that of typicality of AES among the set of sets of states. We are going to adopt the following:

\begin{defi} Consider the bipartition $\mathbb{C}^{d}=\mathbb{C}^{d_1}\otimes\mathbb{C}^{d_2}$. Let $V_{N,AES}(d_1,d_2)$ and $V_{N,TOT}(d)$ be the volume of AES and the total volume, respectively, of $N$-state sets in the space $[\bbC^{d}] ^{\otimes N}$, defined with the Haar measure. Then the relative volume of AES is defined as $V_N(d_1,d_2):=\frac{V_{N,AES}(d_1,d_2)}{V_{N,TOT}(d)}$.
\end{defi}
A straightforward but noteworthy observation is that $V_N$ increases with $N$. Indeed, the relative volume with a larger $N$ is at least as large as that with a smaller one, since if a set of $N_1$ states forms an AES, then any set of $N_2>N_1$ states containing these $N_1$ states would also form an AES. This accounts for the ease of generation of AES with more states. Another interesting observation is that $V_N>0$ if at least one set of $N$ states is known to be AES. This is because that if there exists an AES of $N$ states $\{\ket{\phi_{i}}\}$, one can always show that all $N$-state sets $\{\ket{\psi_{i}}\}$, where each state $\ket{\psi_{i}}$ is in a certain neighbourhood of the state $\ket{\phi_{i}}$, also form AES.  Moreover, we have the following


\begin{theorem}\label{thm3new}
For a bipartition $\mathbb{C}^{d}=\mathbb{C}^{d_1}\otimes\mathbb{C}^{d_2}$, the relative volume $V_N(d_1,d_2)=1$ for any $N> \mathcal{N}(d_1,d_2)$, where
\[
\mathcal{N}(d_1,d_2):=\left\lfloor{\frac{(d_{1}+1)(d_{2}+1)}{2}}\right \rfloor.
\] 
In other words, in a generic set of $N$-states, at most $\mathcal{N}(d_1,d_2)$ states may be simultaneously transformed into product states by a global unitary matrix. Sets of $N$-states violating this property are of zero measure among all set of $N$-states distributed according to the Haar measure.
\end{theorem}

\textit{Proof.} The proof is based on parameter counting. A state $\ket{\phi}=a_{1}\ket{11}+a_{2}\ket{12}...+a_{d}\ket{d_1d_2}$ can be made separable by the action of a unitary matrix $U$ only if the $a'_{j}=\sum_{i=1}^{d}U_{ji}a_{i}$ satisfy
\begin{align}\label{conditionsep2}
a'_{1}a'_{nd_{2}+k}=a'_{k}a'_{nd_{2}+1}  
\end{align}
for every $n=1,2,...,d_1-1$, $k=2,...,d_2$. In other words, separability enforces $(d_1-1)(d_2-1)$ independent complex constraints on each state; and so, separability of the whole set enforces $2N(d_1-1)(d_2-1)$ real constraints. 

In order to generate a quantum state distributed according to the Haar measure, one can use the following recipe: choose $d$ independent complex numbers $a_i$, $1\leq i\leq d$ according to a normal distribution each, then normalise the vector \cite{bengtsson2017geometry}. 
We shall repeat the procedure $N$ times in order to obtain a random set of $N$ qudit states. 
Since the $d$ complex coefficients are drawn independently, the constraints \eqref{conditionsep2} are all independent, and generally require $2N(d_1-1)(d_2-1)$ free parameters to satisfy, up to some special sets of states of zero measure.
Besides, additional local unitary transformations $U_{1}\otimes U_{2}$ do not affect the constraints. 
We can hence define the equivalence relation $R$: two matrices $X,Y\in SU(d)$ are related $X\stackrel{R}{\sim}Y$ if they are locally unitary equivalent, i.e., there are $U_{1}\in SU(d_1)$, and $U_{2}\in SU(d_2)$ such that $Y=(U_{1}\otimes U_{2})X$. 
Then the free parameters available in the quotient set $SU(d)/R$ are $(d_1^2-1)(d_2^2-1)$ real numbers. 
Note that for any quantum state $\ket{\psi}$, and two related matrices $X\stackrel{R}{\sim}Y$, the state $X \ket{\psi}$ meets presented constraints, iff the state $Y \ket{\psi}$ does so. 
Therefore, these many free parameters can only satisfy the constraints enforced by at most $\mathcal{N}(d_1,d_2)$ states. \qed


The smallest case that does not fall under Theorem \ref{thm3new} is $d_1=d_2=2$. First notice that  $\mathcal{N}(2,2)=4$: the construction given in Theorem \ref{thm2} is more specific than a bound based on counting. We proceed to a numerical estimate of $V_{4}(2,2)$. We generate Haar-random sets of states by generating four independent complex numbers according to a normal distribution as the components for each state vector. Then we normalize the vector \cite{bengtsson2017geometry}. By applying the sufficient criterion for AES provided in Theorem \ref{thm1}, we obtain a lower bound in the estimation of $V_4(2,2)$. Specifically,
upon generating $10^7$ sets of states, criterion \eqref{eqcriterion1} run on all 24 permutations of each set detected 2203 AES, i.e.~a fraction $f=2.2\times 10^{-4}$. Since we are studying a binary variable (to be or not to be AES), the standard deviation is $\sigma=\sqrt{10^{7}\times f(1-f) }=47$ (2\% of events). For the upper bound, we minimize the entanglement entropy with Matlab's \texttt{fminunc}. Since running \texttt{fminunc} is considerably slower, we sampled from $10^5$ Haar-random sets of states. We first run criterion \eqref{eqcriterion1}; when this one fails to detect AES, we turn to \texttt{fminunc} to minimise the total entanglement entropy over global unitaries \footnote{We set the threshold for detectable entanglement to $10^{-11}$, since the entropy formula was yielding positive results at the level $10^{-12}$ when fed with product states. A posteriori, we checked that only $0.2\%$ of the 8500 detected cases had an entropy in the range $10^{-8}$ to $10^{-11}$.}. After this procedure, 8500 sets out of $10^5$ were classified as AES, resulting in the proportion $8.5 \times 10^{-2}$, with a standard deviation of 88 ($1\%$ of events). In conclusion, for $V_4$ we have the numerical estimate
\ba
2.2\times 10^{-4}&\lesssim V_4(2,2)&\lesssim 8.5\times 10^{-2}
\ea with error estimates of $1-2\%$ on either bound. 
We conjecture that the true value is closer to the upper bound, as \texttt{fminunc} is quite reliable with problems of this size, while the criterion of Theorem \ref{thm1} may not be tight. 
We have also run the upper bound heuristic estimation for other cases, and we find $V_5(2,2)=V_8(3,3)=1$ (as expected from Theorem \ref{thm3new}) and $V_7(3,3)\approx 0.40$.

To conclude, we give a summary of $V_{N}(d_1, d_2 )$ for bipartite systems, based on the results in this subsection and in \cite{cai2020entanglement}: 
\begin{align}
\begin{split}
&V_N (d_1, d_2 ) =0,\ \textrm{if}\ N\leq \max(d_1,d_2)+1,\\
&V_N (d_1, d_2 ) >0,\ \textrm{if}\ N\geq d_1 +d_2,\\
&V_N (d_1, d_2 ) =1,\ \textrm{if}\ N> \left\lfloor{\frac{(d_{1}+1)(d_{2}+1)}{2}}\right \rfloor
\end{split}
\end{align}
The question on values of $V_N (d_1, d_2 )$ for $\max(d_1,d_2)+1<N<d_1 +d_2$ and $d_1,d_2>2$ remains open; we leave it for future work.


\section{Absolutely entangled set with respect to multi-partition}
\label{sec:multi}

In this section, we generalize the notion of AES to multipartitions of a system and give a lower bound on the number of states in a multipartite AES. Then, we present examples of AES for fixed arbitrary multipartition and construct sets that are absolutely entangled with respect to all partitions of the system.

\subsection{The definition and the lower bound of the number of states.}

Consider any number $d$ with the prime factorization $d=p_1p_2 \cdots p_l$. One may consider bipartitions of the space $\mathbb{C}^d$ according to different factorizations of a number $d$. Moreover, we may consider multipartitions of the space $\mathbb{C}^d$. We begin with the following two definitions:
 
\begin{defi}
\label{Defi3}
Consider a set of quantum states $S=\{\rho_1,\ldots,\rho_N\}$ in a Hilbert space $\mathbb{C}^d$ of non-prime dimension~$d$.
\begin{itemize}
    \item[(i)] The set $S$ is said to be \emph{absolutely entangled with respect to the $(d_1,d_2,\ldots,d_k)$-partition} where $\prod_{i=1}^k d_i =d$, if for every unitary matrix $U\in SU(d)$, at least one state $U\rho_nU^\dagger$ is entangled with respect to multipartition $\mathbb{C}^{d_1}\otimes\cdots\otimes\mathbb{C}^{d_k}$.
    \item[(ii)] Let P be the set of all possible partitions of $\mathbb{C}^d$. The set is said to be \emph{absolutely entangled with respect to any partition}, if for every partition $p\in P$ and every unitary matrix $U \in SU(d)$, at least one state $U\rho_{n}U^\dagger$ is entangled with respect to $p$.
\end{itemize}
\end{defi}

We start by deriving a lower bound on the number of pure states required to form an AES.

\begin{theorem}\label{thm3text} Consider a $k$-partition $\mathbb{C}^d=\mathbb{C}^{d_1} \otimes\cdots\otimes\mathbb{C}^{d_k}$ of the Hilbert space $\mathbb{C}^d$, and denote $d'= \max{(d_1,d_2, \ldots ,d_k)}$. Any set consisting of $d'+1$ vectors cannot form AES with respect to the presented $(d_1, \ldots ,d_k)$-partition.
\end{theorem}

\textit{Proof.} The proof is analogous to that of bipartite states in Ref.~\cite{cai2020entanglement}. 
Consider any set of $d'$ states $\{\ket{\phi_i},i=1,\ldots,d'+1\}$. There exists a basis $\{\ket{\xi_j},j=1,\ldots,d\}$ such that $\ket{\phi_i}=\sum_{j=1}^{d'} c_{ji}\ket{\xi_i}$ for $i=1,\ldots,d'$ and
\ban
\ket{\phi_{d'+1}}&=&\sum_{j=1}^{d'} c_{j,d'+1}\ket{\xi_j}+\beta\ket{\xi_{d'+1}}\equiv\alpha\ket{\Psi}+\beta\ket{\xi_{d'+1}}
\ean where $\ket{\Psi}$ is normalised (note that if all $c_{j,d'+1}=0$, $\ket{\Psi}$ can be any state). Then, assuming $d'=d_1$ without loss of generality, a global unitary can certainly map $\ket{\xi_j}\rightarrow\, \ket{j}\otimes\ket{1}^{\otimes k-1}$ for $j=1,\ldots,d'$; this will induce
\ba
\ket{\phi_i}\longrightarrow\, \ket{\varphi_i}\otimes\ket{1}^{\otimes k-1}&\,,\;&i=1,\ldots,d',
\ea 
and $\ket{\Psi}\rightarrow\ket{\psi}\otimes\ket{1}^{\otimes k-1}$. Now we can choose to map $\ket{\xi_{d'+1}}\rightarrow\, \ket{\psi}\otimes\ket{2}^{\otimes k-1}$, which is indeed orthogonal to all the $\ket{j}\otimes\ket{1}^{\otimes k-1}$; therefore,
\ba
\ket{\phi_{d'+1}}&\longrightarrow& 
\ket{\psi}\otimes(\alpha\ket{1}^{\otimes k-1}+\beta\ket{2}^{\otimes k-1})\,.
\ea Thus, there exists a basis in which all the $d'+1$ states are product, and hence cannot form an AES. \qed

Notice $\max{(d_1,d_2, \cdots ,d_k)}\geq \max_j{(p_j)}$, the largest prime factor of $d$. Thus, given $d$, a set of pure states must contain at least $\max_j{(p_j)}+2$ states to be AES according to some partition. We notice that the same result has also been obtained in Ref.~\cite{lovitz2019decomposable} by considering the decomposition of Gram matrix.

\subsection{Construction of a special family of AES.}

Here we present examples of families of sets that are AES with respect to multipartitions. 

\begin{theorem}\label{thm4}
Consider a Hilbert space $\mathbb{C}^d$ of non-prime dimension with the $k$-partition $\mathbb{C}^d=\mathbb{C}^{d_1} \otimes \mathbb{C}^{d_2}\otimes\cdots\otimes\mathbb{C}^{d_k}$ and an orthonormal basis $\{\ket{\xi_i}\}_{i=1}^d$. Then, the following $N= d_1+\cdots+d_k-k+2$ states,

\begin{align}\label{theset}
\begin{split}
    \ket{\phi_1} &= \ket{\xi_1}, \\
    \ket{\phi_i} &= a\ket{\xi_{1}}+\sqrt{1-a^2}\ket{\xi_{i}},\;i=2,\ldots,N
\end{split}
\end{align}
form an AES for any parameter $a\in(a_{min},1)$, where
\ba
a_{\text{min}}&:=&\max_{i}\,\sqrt{\frac{(d_{i}-1)D_{i}}{(d_{i}-1+\frac{1}{k-1})(D_{i}+\frac{1}{k-1})}},
\ea $i=1,2,\ldots,k-1$ and $D_{i}=\sum_{j=i+1}^{k}(d_{j}-1)+\frac{k-1-i}{k-1}$.
\end{theorem}
\noindent Note that $a_{\text{min}}=\sqrt{\frac{(d_{1}-1)(d_{2}-1)}{d_1d_2}}$ in case of bipartitions, hence recovering the result in Ref.~\cite{cai2020entanglement}.   

Also, Theorem 3 in Ref.~\cite{lovitz2019decomposable} provides a construction of a set of $N'=2\max(d_{1},d_{2},...,d_k)+1$ states which form an AES with respect to the same k-partition $\mathbb{C}^d=\mathbb{C}^{d_1} \otimes \mathbb{C}^{d_2}\otimes\cdots\otimes\mathbb{C}^{d_k}$. So the minimum number of states known to form an AES is currently $\min(N,N')$. Finally notice that $N$ matches the lower bound $d'+2$ of Theorem \ref{thm3text} only for a bipartition ($k=2$) and for $d_2=2$; while $N'>d'+2$ in all cases. Thus the tightness of that lower bound for any multipartition remains an open problem.

Figure \ref{fig3} presents values of the number of states $N$ and parameter $a_{\text{min}}$ for several partitions of the Hilbert space $\mathbb{C}^{32}$. 

\begin{figure}
\renewcommand{\arraystretch}{1.4}
\begin{tabular}{|c|c|c|c|}
\hline
$\:$\specialcell{Partition\\{\small $(d_1,\ldots,d_k )$}}$\:$
&$\:$\specialcell{Partition\\of $\mathbb{C}^d$}$\:$
&$\:$\specialcell{Number of \\ states $N_p$}  $\:$
&$\:$\specialcell{Parameter\\$a_{\text{min},\, p}$ }$\:$
\\ \hline\hline
$(2,16)$& $\mathbb{C}^2 \otimes \mathbb{C}^{16} $ & $ 18$& $0.685$
\\ \hline
$(4,8)$& $\mathbb{C}^4 \otimes \mathbb{C}^{8}$ & $ 12$& $0.810$
\\ \hline\hline
$(2,2,8)$& $(\mathbb{C}^2 )^{\otimes 2} \otimes \mathbb{C}^8$ & $ 11$& $0.789$
\\ \hline
$(2,4,4)$& $ \mathbb{C}^2 \otimes (\mathbb{C}^4 )^{\otimes 2} $ & $ 9$& $0.787$
\\ \hline\hline
$(2,2,2,4)$& $(\mathbb{C}^2 )^{\otimes 3} \otimes \mathbb{C}^4$ & $ 8$& $0.823$
\\ \hline\hline
$(2,2,2,2,2)$& $(\mathbb{C}^2 )^{\otimes 5} $ & $ 7$& $0.800$
\\ \hline\hline
\multicolumn{2}{|c|}{All partitions} & $ 18$& $0.823$
\\ \hline
\end{tabular}
\caption{There are six non-equivalent partitions $p \in P$ of the space $\mathbb{C}^{32}$ related to different factorizations of the number~$32$. 
For each such partition $p=(d_1,\ldots,d_k )$, where $d_1\cdots d_k=32$, we constructed an AES with respect to $(d_1,\ldots,d_k )$-partition (see Definition \ref{Defi3}) consisting of $N_p$ states with the $a_{\text{min},p}$ parameter determined in Theorem \ref{thm4}. 
The set of $N=\max_{p\in P} N_p$ states with the $A= \max_{p\in P} a_{\text{min},p}$ parameter form an AES with respect to any partition.}
\label{fig3}
\end{figure}

\textit{Proof}. The proof of Theorem \ref{thm4} consists of two parts. Firstly, we consider the bipartition $\bbC^d=\bbC^{d_1}\otimes\bbC^{d_{r_1}}$ of the space $\bbC^{d_1 \cdot d_{r_1}}$ and derive values for the parameter $a_{\text{min}}$ sufficient to form AES with respect to given bipartition. Secondly, we apply our approach iteratively, and generalize results for multipartitions.

Firstly, consider the bipartition $\bbC^d=\bbC^{d_1}\otimes\bbC^{d_{r_1}}$, where $d_{r_1}=d_2\cdot d_3\cdots d_k$. Any pure state $\ket{\psi}\in \bbC^{d}$ can be written in the computational basis as
\begin{equation}\label{mteq1}
\ket{\psi}=a_{1}\ket{11}+a_{2}\ket{12}+...+a_{d}\ket{d_1d_{r_1}}.
\end{equation}
We know that $\ket{\psi}$ is separable only if

\begin{equation}\label{mteq2}
a_{1}\cdot a_{n\cdot d_{r_1}+j}=a_{n\cdot d_{r_1}+1}\cdot a_{j},
\end{equation}
where integers $n\in[1,d_1-1]$, $j\in[2,d_{r_1}]$. Taking $|\cdot|^2$ on both sides of Eq.~\eqref{mteq2} and summing over $n$ and $j$ we get, 
\begin{equation}\label{mteq3}
|a_{1}|^{2}\cdot (\sum_{n=1}^{d_1-1}\sum_{j=2}^{d_{r_1}}|a_{n\cdot d_{r_1}+j}|^{2})=\sum_{n=1}^{d_1-1}|a_{n\cdot d_{r_1}+1}|^2\cdot \sum_{j=2}^{d_{r_1}}|a_{j}|^{2}.
\end{equation}


Suppose there exists a unitary matrix $U$ that maps all $N$ states to product states. Without loss of generality, we may assume that $U\ket{\phi_1} = U\ket{\xi_1} = \ket{11}$. 
Remaining states $\{\ket{\xi_i}\}_{i>1}$ of the orthonormal basis transform as
\begin{equation}
U\ket{\xi_i}=b_{i2}\ket{12}+b_{i3}\ket{13}+...+b_{id}\ket{d_1d_{r_{1}}}.
\end{equation}
Since $\{\ket{\xi_i}\}_{i=1}^d$ is an orthonormal basis, the coefficients $b_{ij}$ considered as $(i-1,j-1)$-th entry of a $(d-1,d-1)$ matrix form a unitary matrix, denoted $U_s$. Denote by~$M_N$ a submatrix of first $N-1$ rows of~$U_s$, note that~$M_N$ is a $(N-1,d-1)$ matrix. 
Furthermore, we divide~$M_N$ into three parts: the first part consists of elements in columns $1$ to $d_{r_{1}}-1$, second consists of elements in the columns $n\cdot d_{r_{1}}$ $(n=1,2,\ldots,d_1-1)$, and the third contains the remaining elements. 
Denote the sum of the squared norm of the elements in each part by $S^{(1)}$, $B^{(1)}$, $T^{(1)}$ respectively. 
Similarly, denote by $I$ such a sum of all elements in the matrix $M_N$, see Fig.~\ref{fig1} for a pictorial view of each part. 
In Appendix \ref{supthm3} we show that using Eq.~\eqref{mteq3}, we arrive at a necessary condition of separability, which is,
\begin{equation}\label{mteq4}
     a\leq \sqrt{\frac{(d_1-1)S^{(1)}}{[I-(d_1-1)](I-S^{(1)})}}.    
\end{equation}
The RHS of Ineq.\eqref{mteq4} monotonically increases with $S^{(1)}$, and when $S^{(1)}=D_{1}$ it is $\sqrt{\frac{(d_{1}-1)D_{1}}{(D_{1}+\frac{1}{k-1})(d_{1}-1+\frac{1}{k-1})}}$ (note that $I=N-1$). Therefore if we let $a>\sqrt{\frac{(d_{1}-1)D_{1}}{(D_{1}+\frac{1}{k-1})(d_{1}-1+\frac{1}{k-1})}}$, condition \eqref{mteq4} is satisfied only when $S^{(1)}>D_{1}$.\vspace{\baselineskip}

Secondly, consider further partitioning of subsystem $\bbC^{d_{r_1}}$ into $\bbC^{d_{2}}\otimes\bbC^{d_{r_{2}}}$, where $d_{r_{2}}=d_{3}\cdots d_k$. There is an additional necessary criterion for separability of a given pure state (\ref{mteq1})
\begin{equation}\label{sepstronger}
a_{1}\cdot a_{m\cdot d_{r_2}+j}=a_{m\cdot d_{r_2}+1}\cdot a_{j},
\end{equation}
for all $m\in[1,d_2-1]$, $j\in[2,d_{r_2}]$. Once again, by taking the square norm on both sides of Eq.~\eqref{sepstronger} and summing over $m,j$, the separability criterion takes the form,  
\begin{equation}\label{mteq5}
|a_{1}|^{2}\cdot (\sum_{m=1}^{d_2-1}\sum_{j=2}^{d_{r_2}}|a_{m\cdot d_{r_2}+j}|^{2})=\sum_{m=1}^{d_2-1}|a_{m\cdot d_{r_2}+1}|^2\cdot \sum_{j=2}^{d_{r_2}}|a_{j}|^{2}.
\end{equation}
\begin{figure}[h]
\centering
\includegraphics[scale=0.3]{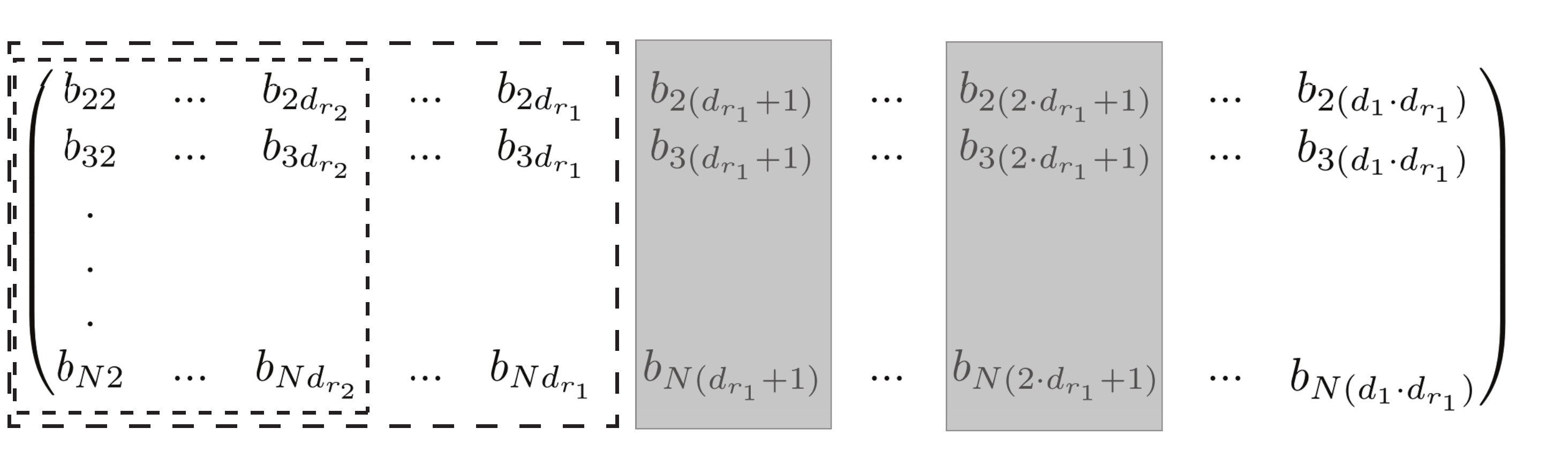}
\caption{Matrix $M_{N}$ divided into three parts: outer block of dashed line related to $S^{(1)}$, shaded blocks to $B^{(1)}$, and the remaining part to $T^{(1)}$ respectively. Note that the shaded part $B^{(1)}$ consists of a number of separate single columns, and when $d_1=2$, there exists only one column of elements $b_{2(d_{r_1}+1)}$ related to $B^{(1)}$. Moreover, the block corresponding to $S^{(1)}$ can be further divided into three parts while considering the partition $\bbC^{d_{r_1}}=\bbC^{d_{2}}\otimes\bbC^{d_{r_2}}$. The inner block of dashed line is related to $S^{(2)}$. The corresponding blocks related to $S^{(i)}$, $i>2$ follow similar pattern.}
    \label{fig1}
\end{figure}
Note that the first part of the matrix $M_{N}$ (corresponding to $S^{(1)}$) might be further divided into three parts in a similar fashion as before. 
In particular, the sum of the squared norm of the elements in the first part (i.e., columns $1$ to $d_{r_2}-1$) equal $S^{(2)}$ (see Fig.~\ref{fig1}). The a necessary condition for separability is now,
\begin{equation}\label{mteq6}
     a\leq \sqrt{\frac{(d_2-1)S^{(2)}}{[S^{(1)}-(d_{2}-1)](S^{(1)}-S^{(2)})}}, 
\end{equation}
the derivation of which is analogous to that of Ineq.\eqref{mteq4} and can be found in Appendix \ref{supthm3}. Since $S^{(1)}>D_{1}$, Ineq.~\eqref{mteq6} is satisfied only if
\begin{equation}\label{mteq7}
     a\leq \sqrt{\frac{(d_2-1)S^{(2)}}{[D_{1}-(d_{2}-1)](D_{1}-S^{(2)})}}. 
\end{equation}
Using Ineq.~\eqref{mteq7}, it can be seen that when we let 
\begin{equation}
a>\sqrt{\frac{(d_{2}-1)D_{2}}{(D_{2}+\frac{1}{k-1})(d_{2}-1+\frac{1}{k-1})}},
\end{equation}
the separability condition \eqref{mteq5} can be satisfied only if $S^{(2)}>D_{2}$.
The remaining subsystems are dealt with analogously, that is,
\begin{equation}\label{conseq}
     a\leq \sqrt{\frac{(d_i-1)S^{(i)}}{[S^{(i-1)}-(d_{i}-1)](S^{(i-1)}-S^{(i)})}} 
\end{equation}
for $i=2,3,\ldots,k$. And if 
\begin{equation}
a\in\Bigg(\textrm{max}_{i}\Big\{\sqrt{\frac{(d_{i}-1)D_{i}}{(d_{i}-1+\frac{1}{k-1})(D_{i}+\frac{1}{k-1})}}\Big\},1
\Bigg),
\end{equation}
then we have a set of necessary conditions for total separability, namely,
\begin{equation}\label{necessaryC}
     S^{(i)}>D_{i}.
\end{equation}
With condition \eqref{necessaryC} we obtain that
\begin{equation}
 S^{(k-1)}>d_{k}-1.
\end{equation} 
However, $S^{(k-1)}$ is the sum of the squared norm of $d_{k}-1$ columns of the matrix $M_{N}$, so there must be 
\begin{equation}
     S^{(k-1)}\leq d_{k}-1.
\end{equation}
The contradiction indicates that the necessary conditions for total separability \eqref{conseq} can not be satisfied simultaneously when the number of input states $N\equiv \sum_{i=1}^{k}d_{i}-k+2$ and $a\in(a_{min},1)$.\qed\vspace{\baselineskip}

From the proof of Theorem \ref{thm4}, we may directly derive two useful corollaries.

\begin{coro}
For bipartition $\bbC^{d}=\bbC^{d_1}\otimes \bbC^{d_2}$, the lower bound of a constant $a_{\text{min}}$ such that the set given in Theorem \ref{thm4} forms an AES decreases when the number of input states $N$ increases from $d_1+d_2$ to $d$.
\end{coro}

\begin{proof}
Consider Ineq.~\eqref{mteq4} for $k=2$. Using $I=N-1$, elementary calculation shows that the lower bound decreases from $\sqrt{\frac{(d_1-1)(d_2-1)}{d_1d_2}}$ to $\frac{1}{\sqrt{d}}$ when $N$ increases from $d_1+d_2$ to $d$. Note the agreement with the relation between relative volume and number of states in an AES. 
\end{proof}

We also conclude the existence of AES with respect to any partition of the system, see Definition \ref{Defi3}(ii).

\begin{coro}
Denote by $\{\xi_i \}_{i=1}^d$ an orthonormal basis of~$\bbC^{d}$, and by $P$ the set of all partitions of~$\bbC^{d}$. Furthermore, for any partition $p \in P$, denote by $\{N_p,a_{\text{min},\, p}\}$ the values of $\{N,a_{\text{min}}\}$ determined by Theorem \ref{thm4}. Then the set of $N:=\max_{p\in P} N_p$ states
\begin{align}
\begin{split}
    \ket{\phi_1} &= \ket{\xi_1}, \\
    \ket{\phi_i} &= a\ket{\xi_{1}}+\sqrt{1-a^2}\ket{\xi_{i}},\;i=2,\ldots,N
\end{split}
\end{align} forms an AES with respect to any partition according to Definition \ref{Defi3}, for $a \in (A,1)$, where $A= \max_{p\in P} a_{\text{min},p}$. 
\end{coro}
This result guarantees that, given any $U$ and any partition $P$, there exists at least one entangled state. An interesting open question is whether one can find an AES, such that, for every $U$, there exists a state that is entangled according to all partitions. This would arguably be the strongest form of genuine multipartite entanglement.


\section{Conclusion}
\label{sec:concl}

This paper dealt with absolutely entangled sets (AESs) of states, a notion first introduced in \cite{cai2020entanglement}. We focused on sets of pure states.

We first presented a number of results for bipartitions $\bbC^{d}=\bbC^{d_1}\otimes\bbC^{d_2}$: a sufficient criterion to detect AES for sets of two-qubit states (Theorem \ref{thm1}), an AES with the maximal number of states that remain entangled under any global unitary (Theorem \ref{thm2}), and a characterisation of the typicality of AES according to the Haar measure (Theorem \ref{thm3new}). 

Then we extended the notion of AES to the multipartite regime. Given a k-partition $\bbC^{d}=\bbC^{d_1}\otimes\bbC^{d_2}\otimes...\otimes\bbC^{d_k}$, we proved that one needs at least $\max{(d_1,d_2,...,d_{k})}+2$ for an AES (Theorem \ref{thm3text}). We constructed an AES with $d_1+d_2+...+d_{k}-k+2$ states (Theorem \ref{thm4}). These two bounds coincide only for the bipartite case ($k=2$) and $\bbC^{d}=\bbC^{2}\otimes\bbC^{d'}$. We also gave a construction of an AES with respect to any possible partition of $\bbC^{d}$. 

An interesting open question is to characterize minimal AES, i.e. find the minimal number of states for an AES in $\bbC^{d}$ given a $k$-partition. As just mentioned, such minimal AES are only known in the bipartite case, and with one subsystems being a qubit. A subsequent question is to find minimal AES such that a minimum number of states remains entangled in any reference frame. 

The study of AES for \textit{mixed states} remains mostly open. Clearly, starting from an AES of pure states, and adding noise to each state, leads to an AES with mixed states, given the amount of added noise is not too large (see \cite{cai2020entanglement} for examples). However, an intriguing question is to determine the minimal size of an AES with mixed states. Notably, it might be possible that smaller AES can be constructed using mixed, compared to the case of pure states, for given Hilbert space dimension. In particular, we do not know if the result of Theorem \ref{thm3text} holds for mixed states as well: so far we could not rule out the possibility that, for $d=4$, three or even two mixed states could constitute an AES \footnote{We can nevertheless prove that a set consisting of one pure and one mixed state cannot be AES. This is because the rank 4 and rank 3 maximally mixed states are always separable upon any global unitary \cite{zyczkowski1998volume}, so we can reduce the mixed state to rank 2. Adding the pure state, we are looking at three pure states in total, which can always be made separable.}. 


More generally, it would be interesting to investigate the relevance of AES in the context of quantum information. One area that certainly connects to AES is that of quantum resources theories. In particular, recent works have started developing frameworks for composite quantum resources \cite{Ducuara_2020}, for instance sets of states featuring quantum coherence in every reference frame \cite{Designolle}. 

Also, in earlier works by Fuchs et al. \cite{fuchs2004quantumness,audenaert2003multiplicativity}, a task was proposed where an eavesdropper tries to reproduce an unknown ensemble according to some measurements on the states, and the lower bound of the average fidelity with respect to all ensembles that a Hilbert space can produce is defined as the quantumness of this Hilbert space. It is proved that some ensembles corresponding to AES can have higher quantumness than all those corresponding to non-AES (the lower is the average fidelity, the higher is the quantumness). Thus, the quantification of entanglement in a AES could be related to some fundamental structures of high-dimensional Hilbert space. 

From a broader perspective, our results are also relevant in the context of linear algebra. In the language of linear algebra, the problem of AES reads as follows (for bipartition $d=d_1d_2$): find a $N\times N$ Gram matrix of rank $r\leq d$, that cannot be written as Hadamard product of two Gram matrices of ranks $r_1\leq d_1$ and $r_2\leq d_2$. In the mathematical literature, to the best of our knowledge, only the work of Lovitz in 2019 \cite{lovitz2019decomposable} has raised this question, also inspired by a problem in quantum information. We hope to see, and to contribute to, the development of more powerful mathematical tools to address these questions.

\section*{Acknowledgements}
We thank Benjamin Lovitz, Jean-Daniel Bancal and Jakub Czartowski for discussions. This research is supported by the National Research Foundation and the Ministry of Education, Singapore, under the Research Centres of Excellence programme, and the National Science Center in Poland under the Maestro grant number DEC-2015/18/A/ST2/00274. We also acknowledge financial support from the Swiss National Science Foundation (project 2000021 192244/1 and NCCR SwissMap). 


\bibliography{ref}

\appendix
\section{Proof of Theorem \ref{thm2}}
\label{appthm2}

We start with a Lemma:

\begin{lemma}\label{lemma1}
Suppose $N$ is an arbitrary positive integer, with $N\geq 4$. There exist $N$ points on a unit sphere in $\mathbb{R}^3$ such that any four of them form a tetrahedron.
\end{lemma}

\textit{Proof.} Any three distinct points on a unit sphere in $\mathbb{R}^3$ determine a plane that intersects the sphere along a circle. Any point on the sphere that lies outside the circle, along with the three points on the circle, would form a tetrahedron. Now assume that the claim of the lemma is false; that is, for any set of $N$ points on a unit sphere, there always exist four points such that they do not form a tetrahedron. Such points must necessarily lie on a circle. This implies that once $ \binom{N-1}{3}$ circles (or less) are determined by $N-1$ points, there remains no choice to add an $N^{th}$ point outside these circles. In other words, this means that the sphere is fully covered by these $\binom{N-1}{3}$ circles, which results in a contradiction as a sphere is not a finite union of circles ($N$ is finite).
\qed\vspace{\baselineskip}

We now prove the theorem. According to Lemma \ref{lemma1}, we can always construct N points on the unit sphere such that any 4 of them form a tetrahedron. Then we obtain N unit vectors $\{\mathbf{v}_{i}=(v_{i1},v_{i2},v_{i3})\}$ from the center of the ball to the N points on the sphere, and construct the N states as-
\begin{align}
\ket{\phi_{i}}&=a\ket{\xi_{0}}+\sqrt{1-a^2}(v_{i1}\ket{\xi_{1}}+v_{i2}\ket{\xi_{2}}+v_{i3}\ket{\xi_{3}}),
\end{align}
where $\{\ket{\xi_{i}}\}$ is an orthonormal basis in $\mathbb{C}^{4}$, $a\in(0,1)$. Now we show that any 4 states in this set form an AES when $a$ is large enough. Let the set of 4 arbitrary states be denoted $\{\ket{\phi_i}\}_{i=0,1,2,3}$. Then these states can be written as,
\begin{align}
    \begin{split}
    \ket{\phi_{0}}&=\ket{\xi_{0}}',\\
    \ket{\phi_{1}}&=c_{10}\ket{\xi_{0}}'+c_{11}\ket{\xi_{1}}',\\
    \ket{\phi_{2}}&=c_{20}\ket{\xi_{0}}'+c_{21}\ket{\xi_{1}}'+c_{22}\ket{\xi_{2}}',\\
    \ket{\phi_{3}}&=c_{30}\ket{\xi_{0}}'+c_{31}\ket{\xi_{1}}'+c_{32}\ket{\xi_{2}}'+c_{33}\ket{\xi_{3}}',
    \end{split}
    \end{align}
where $\{\ket{\xi_{i}}'\}$ is some orthonormal basis. Also let $\mathbf{u}_{1}=\mathbf{v}_{1}-\mathbf{v}_{0}$, $\mathbf{u}_{2}=\mathbf{v}_{2}-\mathbf{v}_{0}$, $\mathbf{u}_{3}=\mathbf{v}_{3}-\mathbf{v}_{0}$, (recall that $\mathbf{v}_{i}$ are the corresponding unit vectors of $\phi_{i}$), where $u_{ij}=v_{ij}-v_{0j}$ denotes the $j$th component of vector $\mathbf{u}_{i}$. Let the states $\ket{\mathbf{u}_{i}}$ (with vanishing coefficients for $\ket{\xi_{0}}$) be
\begin{align}\label{defui}
\ket{\mathbf{u}_{i}}=\frac{u_{i1}\ket{\xi_{1}}+u_{i2}\ket{\xi_{2}}+u_{i3}\ket{\xi_{3}}}{\sqrt{u_{i1}^2+u_{i2}^2+u_{i3}^2}},
\end{align}
then they can be written in the form
\begin{align}\label{defui'}
    \begin{split}
\ket{\mathbf{u}_{1}}=&\ket{\mathbf{u}_{1}'}\\
\ket{\mathbf{u}_{2}}=&U_{21} \ket{\mathbf{u}_{1}'}+U_{22} \ket{\mathbf{u}_{2}'},\\
\ket{\mathbf{u}_{3}}=&U_{31} \ket{\mathbf{u}_{1}'}+U_{32} \ket{\mathbf{u}_{2}'}+U_{33} \ket{\mathbf{u}_{3}'},
\end{split}
    \end{align}
where $\ket{\mathbf{u}_{1}'}, \ket{\mathbf{u}_{2}'}, \ket{\mathbf{u}_{3}'}$ are form an orthonormal set, and $U_{ij}=\braket{\mathbf{u}_{j}'|\mathbf{u}_{i}}$. Since $\mathbf{u}_{1}$, $\mathbf{u}_{2}$, $\mathbf{u}_{3}$ form a linearly independent set (they correspond to three edges with a same vertex in a tetrahedron), $U_{22}$ and $U_{33}$ are nonzero. Now we prove that when $a$ is large enough (approaching 1), the value $\frac{c_{ij}}{\sqrt{1-|c_{i0}|^2}}$ $(i>0, j>0)$ approaches $U_{ij}$ (therefore $L$ would be approaching a nonzero value), and the states $\ket{\xi_{i}}'$ approach $\ket{\mathbf{u}_{i}'}$ $(i>0)$. To see this, let $a=\sqrt{1-\delta^2}$ then $\sqrt{1-a^2}=\delta$. We have,
\begin{equation}
c_{10}=\braket{\phi_{0}|\phi_{1}}=1-\delta^2(1-V_{10}),   
\end{equation}
where $V_{10}=(\mathbf{v}_0,\mathbf{v}_1)$ and
\begin{equation}
|c_{11}|=\sqrt{1-|c_{10}|^{2}}=\delta\sqrt{2(1-V_{10})-\delta^2(1-V_{10})^2}.  \end{equation}
We check the form of $\ket{\xi_{1}}'$. With some calculation we have
\begin{align}\label{xi1tilde}
\ket{\tilde{\xi_{1}}}'
&=\ket{\phi_{1}}-c_{10}\ket{\phi_{0}}\\
&=\sqrt{1-\delta^2}[\delta^2(1-V_{10})]\ket{\xi_{0}}
\nonumber \\
&+\delta[u_{11}+v_{01}\delta^2(1-V_{10})])\ket{\xi_{1}}
\nonumber \\
&+\delta[u_{12}+v_{02}\delta^2(1-V_{10})])\ket{\xi_{2}}
\nonumber \\
&+\delta[u_{13}+v_{03}\delta^2(1-V_{10})])\ket{\xi_{3}}
\nonumber 
\end{align}
where the tilde indicates that the state is unnormalized, and the normalization factor is
\begin{align}\label{normxi1}
|\ket{\phi_{1}}-c_{10}\ket{\phi_{0}}|=\delta\cdot\sqrt{u_{11}^2+u_{12}^2+u_{13}^2+O(\delta^2)}.
\end{align}
From Eqs.~\eqref{xi1tilde} and \eqref{normxi1} we can see that the dominating terms of $\ket{\xi_{1}}'$ are $u_{11}$, $u_{12}$ and $u_{13}$ when $\delta$ becomes small. Therefore for any $\epsilon>0$, we can always find a $\delta_{1}>0$ such that when $\delta<\delta_{1}$, 
\begin{align}
|\braket{\mathbf{u}_{1}'|\xi_{1}}'|>1-\epsilon.
\end{align}
Simply speaking, the state $\ket{\xi_{1}}'$ approaches $\ket{\mathbf{u}_{1}'}$ as $\delta$ tends to $0$.
Similarly, $|c_{20}|=\braket{\phi_{0}|\phi_{2}}=1-\delta^2(1-V_{20})$, and
\begin{align}
\frac{c_{21}}{\sqrt{1-|c_{20}|^2}}
&=\frac{\braket{\xi_{1}'|\phi_{2}}}{\sqrt{1-|c_{20}|^2}}
=\frac{\bra{\xi_{1}'}(\ket{\phi_{2}}-c_{20}\ket{\phi_{0}})}{|\ket{\phi_{2}}-c_{20}\ket{\phi_{0}}|}\nonumber \\
=K(\delta)^{-1}&\Bigg(\delta^2(1-\delta^2)(1-V_{10})(1-V_{20})
\nonumber \\
+\sum_{j=1}^3 [u_{2j}&+v_{0j}\delta^2(1-V_{20})]\cdot[u_{1j}+v_{0j}\delta^2(1-V_{10})]\Bigg),
\nonumber 
\end{align}
where 
\[
K(\delta)=\sqrt{(u_{11}^2+u_{12}^2+u_{13}^2)(u_{21}^2+u_{22}^2+u_{23}^2)+O(\delta^2)}.
\] 
Since we are interested in very small values of $\delta$, we can approximate
\begin{equation}
    \frac{c_{21}}{\sqrt{1-|c_{20}|^2}}=\braket{\mathbf{u}_{1}'|\mathbf{u}_{2}}+O(\delta^2)\,.
\end{equation}
More precisely, for any $\epsilon>0$, we can always find a $\delta_{2}>0$ such than when $\delta<\delta_{2}$,
\begin{equation}
|\frac{c_{21}}{\sqrt{1-|c_{20}|^2}}-U_{21}|<\epsilon.
\end{equation}
Similarly, it can be shown that states $\ket{\xi_2}'$, $\ket{\xi_3}'$ approach $\ket{\mathbf{u}_{2}'}$ and $\ket{\mathbf{u}_{3}'}$ respectively, and the terms $\frac{c_{22}}{\sqrt{1-|c_{20}|^2}}$, $\frac{c_{31}}{\sqrt{1-|c_{30}|^2}}$, $\frac{c_{32}}{\sqrt{1-|c_{30}|^2}}$, $\frac{c_{33}}{\sqrt{1-|c_{30}|^2}}$ approach $U_{22}$, $U_{31}$, $U_{32}$, $U_{33}$, respectively. Therefore, for any $\epsilon>0$, we can find $\delta'$ such that when $\delta<\delta'$,
\begin{equation}
|\frac{2}{L+1}-\frac{2}{L'+1}|< \epsilon,
\end{equation}
where 
\begin{align*}
L\;
&=1+\Bigg(\frac{|c_{21}|+\sqrt{1-|c_{20}|^2}}{|c_{22}|}\Bigg)^2 \\
&+\Bigg(\frac{|c_{31}|+|c_{32}|\cdot\frac{|c_{21}|+\sqrt{1-|c_{20}|^2}}{|c_{22}|}+\sqrt{1-|c_{30}|^2})}{|c_{33}|}\Bigg)^2,\\
L'\;
&=1+\Bigg(\frac{|U_{21}|+1}{|U_{22}|}\Bigg)^2\\
&+\Bigg(\frac{|U_{31}|}{|U_{33}|}+\frac{|U_{32}|}{|U_{33}|}\cdot\frac{|U_{21}|+1}{|U_{22}|}+\frac{1}{|U_{33}|}\Bigg)^2.
\end{align*}
\noindent
We can see that $L'$ is independent of $\delta$ ($a$), and since $\mathbf{u}_{i}$ are linearly independent, $|U_{ii}|\neq 0$. $L'$ will therefore be bounded by some constant, and $1-\frac{2}{L'+1}$ will take a constant value smaller than one, say $o$. Choose $\epsilon'>0$ such that $o+\epsilon'<1$, we can always find a bound $\delta'$ such that when $\delta<\delta'$,  $|1-\frac{2}{L+1}-o|< \epsilon'$ and $c=\textrm{min}_{i}\{|c_{i0}|\}=\textrm{min}_{i}|[1-\delta^2\cdot (1-V_{i0})]|>o+\epsilon'$ hold. Then we have
\begin{align}
\textrm{min}_{i}\{|c_{i0}|\}>1-\frac{2}{L+1},
\end{align}
showing that the 4 states form an AES. Since the same can be repeated for every 4 arbitrary states among the $N$ states, we obtain $\binom{N}{4}$ values for $\delta_i$ $(i=1,...,\binom{N}{4})$, and we just need to let $a$ is large enough such that $\sqrt{1-a^2}\leq \textrm{min}_{i}\{\delta_i\}$. Then we can ensure that there are at least $N-3$ entangled states with respect to any unitary operation, since any 4 states must contain at least 1 entangled state.

\section{Supplemental proof for theorem \ref{thm4}.}
\label{supthm3}
In this section we prove that using the separability criterion
\begin{equation}\label{ap3eq1}
|a_{1}|^{2}\cdot (\sum_{n=1}^{d_1-1}\sum_{j=2}^{d_{r_1}}|a_{n\cdot d_{r_1}+j}|^{2})=\sum_{n=1}^{d_1-1}|a_{n\cdot d_r+1}|^2\cdot \sum_{j=2}^{d_{r_1}}|a_{j}|^{2},
\end{equation}
we can obtain the result in Ineq.~\eqref{mteq4} of the main text, which is
\begin{equation}
     a\leq \sqrt{\frac{(d_1-1)S^{(1)}}{[I-(d_1-1)](I-S^{(1)})}}.    
\end{equation}
And we also give more details for the derivation of Ineq.~\eqref{mteq6}.

\textit{Proof.}
Applying separability condition \eqref{ap3eq1} to the $N-1$ transformed states $U\ket{\phi_i}$ $(i=2,3,...,N)$, each with the form 
\begin{equation}
U\ket{\phi_i}=a\ket{11}+\sqrt{1-a^2}(b_{i2}\ket{12}+b_{i3}\ket{13}+...+b_{id}\ket{d_1d_{r_{1}}}),
\end{equation}

we have 
\begin{equation}\label{ap3eq2}
a^{2}\cdot T_{i}^{(1)}=(1-a^{2}) B_{i}^{(1)}\cdot S_{i}^{(1)},
\end{equation}
where $S_{i}^{(1)}=\sum_{j=2}^{d_{r_1}}|b_{ij}|^{2}$, $T_{i}^{(1)}=\sum_{n=1}^{d_1-1}\sum_{j=2}^{d_{r_1}}|b_{i(n\cdot d_{r_1}+j)}|^{2}$, $B_{i}^{(1)}=\sum_{n=1}^{d_1-1} |b_{i(n\cdot d_{r_1}+1)}|^2$.
Since every row and column of $U_{s}$ is normalized, we have
$S_{i}^{(1)}+T_{i}^{(1)}+B_{i}^{(1)}=1$ $(i>1)$, therefore
\begin{equation}\label{ap3eq3}
S_{i}^{(1)}=\frac{a^{2}(1-B_{i}^{(1)})}{B_{i}^{(1)}+a^2(1-B_{i}^{(1)})}. 
\end{equation}
Summing up Eq.~\eqref{ap3eq3} with respect to $i$, we have
\begin{equation}\label{ap3eq4}
S^{(1)}=\sum_{i=2}^{N} S_{i}^{(1)}=\sum_{i=2}^{N}\frac{a^{2}(1-B_{i}^{(1)})}{B_{i}^{(1)}+a^2(1-B_{i}^{(1)})}. 
\end{equation}
Remember that we defined that $B^{(1)}=\sum_{i=2}^{N} B_{i}^{(1)}$. Now using the unitarity of $U_{s}$, we know that $B^{(1)}\leq d_1-1$. Substituting one of the variables $B_{i}^{(1)}$ in \eqref{ap3eq4}, say $B_{N}^{(1)}$, with $B_{N}^{(1)}=B^{(1)}-\sum_{i=2}^{N-1}B_{i}^{(1)}$ and taking the derivatives of $S^{(1)}$ with respect to every other $B_{i}^{(1)}$, with a simple calculation we can see that the lower bound of $S^{(1)}$ is attained if and only if we let $B^{(1)}= d_1-1$ and all $B_{i}^{(1)}$ equal, namely, $B_{i}^{(1)}=\frac{d_{1}-1}{N-1}$ for every $i$. So we have
\begin{equation}\label{ap3eq5}
S^{(1)}\geq\frac{(N-1)(N-d_1)a^{2}}{(d_1-1)(1-a^2)+(N-1)a^2},
\end{equation}
from which we get,
\begin{equation}\label{apeq7}
a\leq\sqrt{\frac{(d_1-1)S^{(1)}}{[N-(S^{(1)}+1)](N-d_1)}}.
\end{equation}
Using the relation $N=I+1$, we obtain the result
\begin{equation}
     a\leq \sqrt{\frac{(d_1-1)S^{(1)}}{[I-(d_1-1)](I-S^{(1)})}}.    
\end{equation}
\qed

Now we give more details for the derivation of Ineq.~\eqref{mteq6}, which takes the form
\begin{equation}
     a\leq \sqrt{\frac{(d_2-1)S^{(2)}}{[S^{(1)}-(d_{2}-1)](S^{(1)}-S^{(2)})}}. 
\end{equation}
The derivation is similar to Ineq.~\eqref{mteq3}. First we have
\begin{equation}\label{ap3eq8}
a^{2}\cdot T_{i}^{(2)}=(1-a^{2}) B_{i}^{(2)}\cdot S_{i}^{(2)},
\end{equation}
where $S_{i}^{(2)}=\sum_{j=2}^{d_{r_2}}|b_{ij}|^{2}$, $T_{i}^{(2)}=\sum_{n=1}^{d_2-1}\sum_{j=2}^{d_{r_2}}|b_{i(n\cdot d_{r_2}+j)}|^{2}$, $B_{i}^{(2)}=\sum_{n=1}^{d_2-1} |b_{i(n\cdot d_{r_2}+1)}|^2$. Notice that the corresponding areas of $S_{i}^{(2)}, T_{i}^{(2)}, B_{i}^{(2)}$ belong to the first $d_{r_1}$ columns of the matrix $M_{N}$, and the three terms $S_{i}^{(2)}, T_{i}^{(2)}, B_{i}^{(2)}$ no longer sum to one. Instead, we have
\begin{equation}\label{ap3eqadd}
S_{i}^{(2)}+T_{i}^{(2)}+B_{i}^{(2)}=L_{i},    
\end{equation}
where $L_{i}\geq0$ and $\sum_{i=2}^{N}L_{i}=S^{(1)}$. Then,
\begin{equation}\label{ap3eq9}
S^{(2)}=\sum_{i=2}^{N} S_{i}^{(2)}=\sum_{i=2}^{N}\frac{a^{2}(L_{i}-B_{i}^{(2)})}{B_{i}^{(2)}+a^2(1-B_{i}^{(2)})}. 
\end{equation}
Using the unitarity of $U_{s}$ we know that $B^{(2)}\equiv\sum_{i=2}^{N}B_{i}^{(2)}\leq d_2-1$. Substituting one of the variables $L_{i}$ in \eqref{ap3eq9}, say $L_{N}$, with $L_{N}=S^{(1)}-\sum_{i=2}^{N-1}L_{i}$, taking the derivatives of $S^{(2)}$ with respect to every other $L_{i}$, with some calculation we can see that the lower bound of $S^{(2)}$ is attained if and only if we let all $B_{i}^{(2)}$ equal, and let $B^{(2)}= d_2-1$, namely, $B_{i}^{(2)}=\frac{d_{2}-1}{N-1}$ for every $i$. Substituting $B_{i}^{(2)}=\frac{d_{2}-1}{N-1}$ into \eqref{ap3eq9}, we have \begin{equation}\label{ap3eq10}
S^{(2)}\geq\frac{(N-1)(S^{(1)}-d_2+1)a^{2}}{(d_2-1)(1-a^2)+(N-1)a^2},
\end{equation}
which is analogous to \eqref{ap3eq5}. With some calculation we have
\begin{equation}\label{ap3eq11}
a^2\leq\frac{(d_{2}-1)S^{(2)}}{NP+d_{2}S^{(2)}-S^{(1)}+d_2-1},
\end{equation}
where $P=S^{(1)}-S^{(2)}-(d_2-1)$. We can see $P\geq 0$ by summing up \eqref{ap3eqadd} with respect to $i$. Also, the denominator of  \eqref{ap3eq11} is always positive as long as the integer that P multiplied by is not smaller than $1$, and the right hand side of \eqref{ap3eq11} increases when the integer decreases. Since $N\geq S^{(1)}+1$, by substituting the $N$ with $S^{(1)}+1$ in \eqref{ap3eq11}, we finally obtain Ineq.~\eqref{mteq6} with some simple calculation.

\end{document}